\documentclass{article}

 \usepackage{graphicx}

\usepackage{amsmath, amssymb, amsthm, enumerate}

\usepackage{amsmath,amssymb,amsthm, enumerate,mathrsfs}

\theoremstyle{definition}

\numberwithin{equation}{section}

\usepackage{eepic}
\usepackage{epic}
\usepackage{mathtools}
\usepackage{color}
\usepackage{float}
\usepackage{url} 
\usepackage{array}
\usepackage{xcolor}
\theoremstyle{plain}
\newtheorem{thm}{Theorem}[section]
\newtheorem{crl}[thm]{Corollary}

\newtheorem{prp}[thm]{Proposition}

\newtheorem{problem}{Open problem}
\theoremstyle{definition}
\newtheorem{dfn}[thm]{Definition}
\theoremstyle{remark}
\newcommand{\Mn}{\mathrm{Met}_n}
\newcommand{\Rn}{\mathrm{RMet}_n}
\newcommand{\RG}{\mathrm{RMet}(G)}
\newcommand{\MG}{\mathrm{Met}(G)}
\newcommand{\Cn}{\mathrm{Cut}_n}
\newcommand{\CG}{\mathrm{Cut}(G)}

\newcommand{\En}{\mathcal{E}_n}
\newcommand{\In}{\mathcal{I}_n}
\newcommand{\IG}{\mathcal{I}(G)}
\newcommand{\EG}{\mathcal{E}(G)}
\newcommand{\EK}{\mathcal{E}(K_n)}

\newcommand{\R}{\mathbb{R}}

\newcommand{\Vol}{\mathrm{vol}}

\begin{document}

\begin{center}
\large{\bf On the volumes of the elliptope, metric polytope, and cut polytope}
\end{center}\vspace{5mm}

\begin{center}
\textsc{David Avis, Luc Devroye, and Antoine Deza}
\end{center}

\vspace{2mm}

\footnotesize{
\noindent\begin{minipage}{14cm}
{\bf Abstract:}
In this paper, we investigate the relationships between the volumes of four
convex bodies:
the cut polytope, metric polytope, rooted metric polytope, and
elliptope, defined on graphs with $n$ vertices.
After an affine change of coordinates for the elliptope,
the cut polytope is contained in each of the other three, which,
for optimization purposes, provide polynomial-time relaxations.
It is therefore of interest to see how tight these relaxations are.
Worst-case ratio bounds are well known, but these are limited to
objective functions with non-negative coefficients. Volume ratios,
pioneered by Jon Lee with several co-authors,
give global bounds and are the subject of this paper. 
For the rooted metric polytope over
the complete graph, we show
that for large $n$ its volume is much greater than that of the elliptope.
For the metric polytope, for small
values of $n$, we show that its volume is smaller
than that of the elliptope; however, for large values,
we prove that the converse is true. Volume estimates place the
crossover near $n=13$.
We also give exact formulae for the volumes of several families of sparse cut polytopes.
In particular, we give an exact formula for the volume of the elliptope of every cycle
and show that its volume ratio with the corresponding cut polytope converges rapidly to one.
\end{minipage}
\phantom{i}

\hbox to14cm{\hrulefill}\par
}



\section{Introduction and literature review}
In this paper, we investigate the relationships between four
convex bodies:
the cut polytope, metric polytope, rooted metric polytope and elliptope.
Formal definitions are given in the next section, and a general reference is the book by Deza and Laurent \cite{DL97}.
For volume comparisons we use the affine image of the elliptope
defined in the next section, so that all four bodies are expressed
in the same coordinates.
The metric polytope, rooted metric polytope, and elliptope all contain the cut polytope
and hence, for optimization purposes, provide relaxations. 
The rooted metric polytope contains the metric polytope.
In general, neither is comparable by inclusion with the elliptope.
The metric polytope and rooted metric polytope are linear relaxations defined by the triangle inequalities, all
of which are facets of the cut polytope. The elliptope provides a
semidefinite relaxation that contains all vertices of the cut polytope
as extreme points. Optimization over these three bodies can be performed
in polynomial time, whereas optimization is NP-hard in general for the cut polytope.
The goal of this paper is to compare how well the three relaxations
approximate the cut polytope in terms of their respective
volumes. 

We denote the volume of a closed, bounded convex set $K$ as $\Vol(K)$.
Consider polytope $P \subseteq \mathbb{R}^n$ and closed bounded convex set $Q \subseteq \mathbb{R}^n$
such that $P \subseteq Q$. 
Such pairs of convex sets appear frequently in cases where optimization
over $P$ is NP-hard, but can be done in polynomial time over $Q$.
In this case, $Q$ gives a polynomial-time computable relaxation of
the NP-hard problem.
Two measures of
how closely $Q$ approximates $P$ are
\[
\frac{\Vol(Q)}{\Vol(P)}
~~~~~~\text{and}~~~~~~
\frac{\max\{cx:x\in Q\}}{\max\{cx:x\in P\}}.
\]
The worst-case ratio bound on the right requires objective functions
with non-negative coefficients, to avoid a zero in the denominator,
whereas the volume ratio bound is global.

We will be considering the case where $P=\Cn$, the 
cut polytope, and either $Q=\Mn$ the metric polytope,
$Q=\Rn$, the rooted metric polytope, or $Q=\In$ the elliptope.
For $n=3$, these relationships are shown in Figure \ref{C3}.
\begin{figure}
\centering
\includegraphics{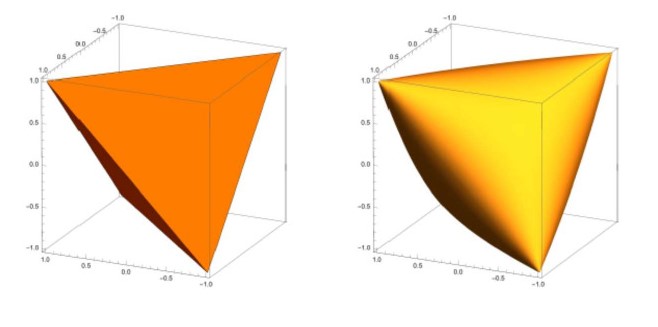}
\caption{$\mathrm{Cut}_3=\mathrm{Met}_3$~~~~~~~~~~~~~~~~~~Elliptope $I_3$}
\label{C3}
\end{figure}
For this special case, $\mathrm{Cut}_3=\mathrm{Met}_3=\mathrm{RMet}_3$ is shown on the left.
The elliptope is shown on the right, and the two bounds above are: 
\[
\frac{\Vol(I_3)}{\Vol(\mathrm{Cut}_3)} \approx 1.851
~~~~~~\text{and}~~~~~~
\frac{\max\{cx:x\in I_3\}}{\max\{cx:x\in \mathrm{Cut}_3\}}
\le 1.125,
\qquad c_j\ge0,\ \forall j,\ c\ne0.
\]

The general 
worst-case ratio bounds with a non-negative objective function are
$2-\epsilon  ~~$ when $Q=\Mn$ (Poljak and Tuza \cite{PT94})
versus
$1.139~$ when $Q=\In$ (Goemans and Williamson \cite{GW95}).
We will be considering how to obtain volumes for these two
bodies in this paper to obtain general volume bounds.
 
The use of volumes to compare combinatorial polytopes, test the strength
of constraints, and study other polyhedral properties was pioneered by Jon Lee
in a series of papers with various coauthors. Lee and Morris \cite{LM94} consider a volume ratio bound similar to that above except scaled
by taking it to the power $1/d$, where $d$ is the dimension.\footnote{ We prefer not to scale by $1/d$ as we can
often compute volumes exactly in rational arithmetic.} Ko, Lee and Steingrimsson
\cite{KLS97} study this bound applied to the 
{\em boolean quadric polytope} (BQP), 
otherwise known as the {\em correlation polytope} and defined 
in Section \ref{def}.
It is related to the cut polytope by a linear bijection known as the covariance map.
They consider a relaxation of BQP to the rooted metric polytope $\Rn$, which is weaker
than the metric polytope as it contains only $O(n^2)$ of the $O(n^3)$
facets defining $\Mn$. They proved the following:
\begin{thm}[Ko, Lee and Steingrimsson \cite{KLS97}]
\label{KLS}
For $n \ge 2$,
\begin{equation*}
\Vol(\mathrm{RMet}_{n+1}) = \frac{2^n n!}{(2n)!},~~~~\mathrm{hence}~~~~\Vol(\mathrm{RMet}_{n+2})= \frac{\Vol(\mathrm{RMet}_{n+1})}{2n+1}.
\end{equation*}
\end{thm}
The work of Lee and Skipper \cite{LS20} and Lee, Skipper and Speakman \cite{LSS18}, who
compute volumes of sparse BQP, is closely related to ours. We give the relevant results in
Section \ref{exact}.
For small values of $n$, it is possible to compute
the volumes of $\Cn$ and $\Mn$ exactly\footnote{Using {\em lrslib v7.3} at
\url{https://cgm.cs.mcgill.ca/~avis/C/lrs.html} and Normaliz \url{https://www.normaliz.uni-osnabrueck.de/}}
and to high precision for $\In$.\footnote{Using the {\em fungible} package
at \url{https://CRAN.R-project.org}}
These results are given in Table \ref{tab:volumes}, where Theorem \ref{KLS} is used for $\Vol(\Rn)$.
The volume of $\Cn$ for $n \le 6$ was previously
computed by Deza, Deza and Fukuda \cite{DDF95}.
In contrast to the worst-case ratio bounds cited above, for $n \le 7$ the
metric polytope gives a much tighter volume relaxation than the elliptope. 
For $n=6,7$, $\Rn$ gives a weaker bound than $\In$.

\begin{table}[htb]
\caption{Volumes and ratios for small $n$.}
\label{tab:volumes}
\setlength{\extrarowheight}{5pt}
\begin{tabular}{||c|c|c|c|c|c|c|c||}
\hline
n &$\Vol(\In)$  & $\Vol(\Cn)$ &$\Vol(\Mn)$& $\Vol(\Rn)$& $\frac{\Vol(I_n)}{\Vol(\Cn)}$&  $\frac{\Vol(\Mn)}{\Vol(\Cn)}$& $\frac{\Vol(\Rn)}{\Vol(\Cn)}$\\
\hline
2 & 1 & 1 & 1 & 1& 1& 1 &1 \\
3 & 0.617    &     1/3       & 1/3&  1/3 &   1.85 &    1 &1 \\
4 & 0.183    &     2/45      & 2/45&  1/15    &       4.11   &  1 & 1.5\\
5 & 0.022   &    32/14175   &     4/1701  &1/105 &         9.75   & 1.04 & 4.22 \\
6 & 9.50e-04 & $\frac{2384}{58046625}$ &    $\frac{71936}{1477701225}$ &1/945&  23.1  &    1.19&25.8 \\
7 & 1.33e-05 & $\frac{637888}{2531123083125}$& $\frac{3586321206047}{9206753675223604500}$&1/10395& 52.7 &      1.55& 382 \\
\hline
\end{tabular}
\end{table}

Our main asymptotic results show a  separation between these
relaxations. We prove that
\[
\Vol(\In)
=
\exp\left(
-\frac{n^2\log n}{4}
+\Theta(n^2)
\right),
\]
and that the cut polytope has the same leading asymptotic behaviour:
\[
\Vol(\Cn)
=
\exp\left(
-\frac{n^2\log n}{4}
+\Theta(n^2)
\right).
\]
For the metric polytope, we prove that there exists a constant
$
\lambda_{\mathrm{Met}}\in[0.7027,\log 3]
$
such that
\[
\log\Vol(\Mn)
=
-\lambda_{\mathrm{Met}}\binom{n}{2}
+o(n^2).
\]
Consequently,
\[
\frac{\Vol(\In)}{\Vol(\Mn)}
\longrightarrow 0,
\]
so asymptotically the elliptope gives a much tighter volume relaxation
than the metric polytope. We also show that the rooted metric polytope
is asymptotically much larger than the elliptope. The volume estimates
in Section \ref{estimation} place the crossover between $\Mn$ and $\In$
near $n=13$.
 
The convex bodies described so far are defined in $\mathbb{R}^{\binom{n}{2}}$,
the coordinates of which can be conveniently labelled as edges of
the complete graph $K_n$. However, it is also of interest to study the projections
of these bodies onto the edges of an arbitrary graph $G$.
In the case of the cut polytope, we can give explicit volume
formulae for several classes of sparse graphs.

The paper is organized as follows.
In Section \ref{def}, we provide formal definitions for the convex bodies.
Then, in Section \ref{exact}, we give exact formulae for
the volume of the cut polytope for various sparse graphs
with no $K_5$-minor and an exact formula for the volume
of the elliptope of a cycle.
Sections \ref{asymptotic} and \ref{asymptotic2} contain the
asymptotic results stated above.
In Section \ref{estimation} we use the Volesti \cite{CF21}
package to estimate the volume of $\Mn$ for $8 \le n \le 25$.
In the final section, we give some general conclusions and state
some open problems.

\section{Definitions}
\label{def}

We begin by defining the convex bodies of interest in this paper.
The standard reference is Deza and Laurent \cite{DL97}.
Unless otherwise stated, let $n \ge 3$, and let $[n]$ denote $\{1,\ldots,n\}$.
The cut polytope is given by a vertex representation, whereas the
elliptope is given by a positive-semidefinite Gram representation.
\noindent
\begin{dfn}[Cut polytope $\CG$ ]
\label{Cdef}

Let $G(V, E)$ be a finite undirected graph with $V=[n]$.
For any $S \subseteq [n-1] $,
we define the
{\em cut vector} $x^S \in \R^{|E|}$ by

\begin{eqnarray*}
x^S_{ij}= \left\{ \begin{array}{ll}
   1 & {~~~~ | \{i,j\} \cap S | = 1 }   \\ 
   0 & {~~~~otherwise}
                    \end{array} \right.
~  \forall (i,j) \in E
\end{eqnarray*}
and
\[
\CG= CH \{x^S : \forall ~ S \subseteq [n-1]\} \subset \R^{|E|}.
\]
We abbreviate $\mathrm{Cut}(K_n)$ to $\Cn$.
\qed
\end{dfn}
\noindent

\begin{dfn}[Elliptope $\EG$]
\noindent
Let $G=(V,E)$ be a graph with $V=[n]$, let $Y=[ y_{ij}]$ be a
$n \times n$ matrix and let $y^E =  (y_{ij} : (i,j) \in E ) \in \R^{|E|}$. 
\[
    \EG=\{y^E : Y=[ y_{ij}],~  n \times n,~\text{symmetric,~positive semidefinite},~ y_{ii}=1, i=1,\ldots,n \} \subset \R^{|E|}
\]
We abbreviate $\EK$ to $\En$.
\qed
\end{dfn}

The metric and rooted metric polytopes are defined by systems of inequalities known
as $H$-representations.
First, let us define them for the complete graph $K_n$.
\begin{dfn}[Metric polytope $\Mn$] 
\label{Mdef}
$\Mn$ 
is the set of $x \in \R^{\binom{n}{2}}$ satisfying the
$4 \binom{n}{3}$ triangle inequalities:
\begin{equation}
\label{triangle}
\begin{aligned}
x_{ij} &\le x_{ik}+x_{jk},\\
x_{ik} &\le x_{ij}+x_{jk},\\
x_{jk} &\le x_{ij}+x_{ik},\\
x_{ij}+x_{ik}+x_{jk} &\le 2
\end{aligned}
\end{equation}
for all $1 \le i < j < k \le n$
\qed
\end{dfn}   
\noindent
\begin{dfn}[Rooted metric polytope $\Rn$]
\label{RMdef}
$\Rn$
is the set of $x \in \R^{\binom{n}{2}}$ satisfying the
$4 \binom{n-1}{2}$ triangle inequalities
\[
\begin{aligned}
x_{ij} &\le x_{in}+x_{jn},\\
x_{in} &\le x_{ij}+x_{jn},\\
x_{jn} &\le x_{ij}+x_{in},\\
x_{ij}+x_{in}+x_{jn} &\le 2
\end{aligned}
\]
for all $1 \le i < j \le n-1$
\qed
\end{dfn}
\noindent
$\mathrm{Met}_2$ and $\mathrm{RMet}_2$ are defined by the segment
$0 \le x_{12} \le 1$.
Given a graph $G=(V,E)$ on $n \ge 2$ vertices we denote the projection
of $\Mn$ (resp., $\Rn$) onto the variables indexed by $E$ as
$\MG$ (resp., $\RG$).
The cut vectors defining $\CG$ are vertices of both $\MG$ and $\RG$.
They are the only integer vertices of $\MG$.
For a direct comparison of $\EG$ (expressed in coordinates $y_{ij}$)
with the cut and metric
polytopes (expressed in coordinates $x_{ij}$) we transform $\EG$
into $\IG$ by the mapping
\begin{equation}
\label{E2I}
x_{ij}=\frac{1-y_{ij}}{2},~~~~(i,j)\in E.
\end{equation}
This implies that
\begin{equation}
\label{VE2I}
\Vol (\EG) = 2^{|E|} \Vol(\IG ).
\end{equation}
We will also use the coordinatewise angular map
\begin{equation}
\label{angularmap}
\Phi_G(x)_{ij}
=
\frac{1}{\pi}\arccos(1-2x_{ij}),\qquad (i,j)\in E.
\end{equation}
The map is injective, with inverse on its image
\[
\Phi_G^{-1}(a)_{ij}
=
\frac{1-\cos(\pi a_{ij})}{2}.
\]
$\In=\mathcal{I}(K_n)$ has an $H$-representation, albeit with
a countable number of inequalities (see \cite{DL97}, p. 459):
\begin{eqnarray*}
\In = \left \{ x \in \mathbb{R}^{\binom{n}{2}}~\Bigg|~     \sum_{1 \le i < j \le n} b_i b_j x_{ij} \le \frac{1}{4} \left(\sum_{i=1}^n b_i \right) ^2~~\forall b \in \mathbb{Z}^n \right\} .
\end{eqnarray*}
An $H$-representation of $\IG$ can be obtained from this by projection.

It is known that $\CG  \subset \IG$. Although $\MG$
and  $\IG$ are not comparable in general, their intersection 
contains $\CG$.
All of these convex bodies
are full-dimensional.

We conclude this section by introducing a bijection of
the cut polytope, which we will need in the next section.
\begin{dfn}[Correlation polytope $\text{Cor}(G)$]
Let $G=(V, E)$ be a finite undirected graph with $n$ vertices.
For $S \subseteq [n]$, the {\em correlation vector}
\begin{equation*}
y^S=\{y_{11},y_{22},\ldots,y_{nn},~y_{ij}~ (i,j)\in E \} \in \R^{|V|+|E|}
\end{equation*}
is given by:
\begin{eqnarray*}
~~~~~~ y^S_{ij}= \left\{ \begin{array}{ll}
   1 &  i=j \in S ~~or~~ (i,j) \in E,~\{i,j\} \subseteq S,\\
   0 & \text{otherwise,}
                    \end{array} \right.
\end{eqnarray*}
and 
\[
\text{Cor}(G)= CH \{y^S : \forall S \subseteq [n]~\}~\subset~\R^{|V|+|E|}.
\]
\end{dfn} 

The {\em suspension} $\nabla(G)$ of a graph $G$ with $n$
nodes is formed by adding vertex $n+1$
and edges $i,n+1$ for $i \in V$ as illustrated in Figure \ref{susp}.

\begin{figure} [H]
\centering
\includegraphics[scale=0.4]{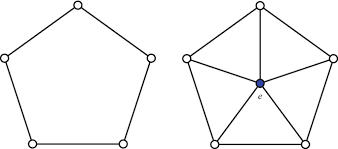}
\caption{The cycle $C_5$ and its suspension $\nabla(C_5)=W_6$.}
\label{susp}
\end{figure}

The cut and correlation polytopes are related by the 
covariance map: 
\[
\begin{aligned}
 x_{i,n+1}&=y_{ii} &&1 \le i \le n,\\
 x_{ij}&=y_{ii}+y_{jj}-2y_{ij}&&(i,j)\in E.
\end{aligned}
\]

It is easy to show (e.g., see \cite{DL97}, Section 5.2) that
\[
\text{Cor}(G) \cong \mathrm{Cut}(\nabla(G)).
\]
Moreover, the covariance map has determinant of absolute value
$2^{|E|}$, and hence
\[
\Vol(\mathrm{Cut}(\nabla G))
=2^{|E|}\Vol(\mathrm{Cor}(G)).
\]

\section{Exact volumes for
some graphs with no $K_5 $ minor}
\label{exact}

We will need the following result, which was proved 
for the cut cone by Seymour, and later extended to the cut
polytope by Barahona and Mahjoub (\cite{DL97}, Theorem 27.3.6).
\begin{thm}
\label{seymour}
$\CG = \MG$ if and only if $G$ contains no $K_5$ minor. \\
In this case the polytope is described by the box inequalities
$0\le x_e\le 1$ and the cycle inequalities
\begin{equation*}
\sum_{e \in F} x_e -\sum_{e\in C \setminus F} x_e
\leq |F|-1,~~~|F|~odd, F \subseteq C
\end{equation*}
for every induced cycle $C$ in $G$.
\qed
\end{thm}
\noindent
If $|C|=3$, then these cycle inequalities are just the triangle
inequalities (\ref{triangle}).

The {\em (prism) product} of bounded convex sets $P \subseteq \R^r$
and $Q \subseteq \R^s$ 
is given by
\[
P \times Q =  \{ (p,q): p \in P, q\in Q \} \subseteq{\R^{r+s}}.
\]
It is well-known that 
\begin{equation}
\label{volprod}
\Vol(P \times Q) = \Vol(P) \Vol(Q).
\end{equation}
Further if $P$ and $Q$ are polytopes
with vertex sets $V(P)$ and $V(Q)$, respectively, 
then
\begin{equation}
\label{vprod}
P \times Q = CH \left\{ \binom{x}{y} : x \in V(P), y \in V(Q) \right\}.
\end{equation}
\begin{prp}
\label{union}
If graphs $G_1$ and $G_2$ are disjoint or
share one vertex, then
\begin{equation}
\label{uprod}
\mathrm{Cut}(G_1 \cup G_2 ) = \mathrm{Cut}(G_1) \times \mathrm{Cut} (G_2) \\ 
\end{equation}
and
\begin{equation}
\label{uvolprod}
\Vol(\mathrm{Cut}(G_1 \cup G_2))
=
\Vol(\mathrm{Cut}(G_1)) \Vol(\mathrm{Cut}(G_2)).
\end{equation}
\end{prp}
\begin{proof}
Let $G=G_1 \cup G_2$ and suppose it has $n$ vertices. If $G_1$ and $G_2$ share a vertex,
label this vertex $n$. Consider any cut $S \subseteq [n-1]$, then it
decomposes into disjoint cutsets $S_1$ for $G_1$ and $S_2$ for $G_2$.
The cut vector $x^S$ decomposes into  cut vectors $x^{S_1}$ for $G_1$ and
$x^{S_2}$ for $G_2$. Conversely 
concatenating cut vectors $x^{S_1}$ for $G_1$ and
$x^{S_2}$ for $G_2$ give a cut vector $x^S$ for $G$.
Therefore (\ref{uprod}) follows from (\ref{vprod})
and so (\ref{uvolprod}) follows from (\ref{volprod}).
\end{proof}

\begin{prp}
\label{for}
If $F$ is a forest then $\Vol(\mathrm{Cut}(F))=1$.
\end{prp}
\begin{proof}
$\mathrm{Cut}(K_2)$ is the line segment with endpoints 0 and 1
and therefore has volume (length) 1.
Any tree $T$ can be constructed from $K_2$ by successively
adjoining a new edge at an existing vertex.
It follows from Proposition \ref{union} that
$\Vol(\mathrm{Cut}(T))=1$. Since $F$ is a collection of trees, it again follows from Proposition \ref{union} that
$\Vol(\mathrm{Cut}(F)) = 1$.
\end{proof}

Next, we compute the volume of the cycle $C_n$ on $n$ vertices.
\begin{prp}
\label{cycle}
\begin{equation*}
\Vol(\mathrm{Cut}(C_n))= 1 - \frac{2^{n-1}}{n!}.
\end{equation*}
\end{prp}
\begin{proof}
By Theorem \ref{seymour} there are $2^{n-1}$ non-box cycle
inequalities for $\mathrm{Met}(C_n)$ given by:
\begin{equation*}
\sum_{e \in F} x_e -\sum_{e\in C \setminus F} x_e \leq |F|-1,~~~|F|~odd, F \subseteq C.
\end{equation*}
Together with the box inequalities, these describe $\mathrm{Cut}(C_n)$.
$\mathrm{Cut}(C_n) \subset [0,1]^n$ and has $2^{n-1}$ vertices.
Each cycle inequality cuts off the cube vertex
$x_e=1, e \in F$, $x_e=0,e \in C \setminus F$.
The $2^{n-1}$ simplices have disjoint interiors,
each with volume $1/n!$ and so the volume of the cube is reduced from
1 to $ 1 - \frac{2^{n-1}}{n!}$. This remaining set is $\mathrm{Cut}(C_n)$
since we have intersected the cube
with all of the cycle inequalities. 
\end{proof}

We now combine these previous results to compute the volumes
of additional classes of graphs. 
A {\em cactus} is a connected graph where each edge is in at most one cycle.
An example is shown in Figure \ref{cactus}.

\begin{figure} [H]
\centering
\includegraphics[scale=0.3,angle=90]{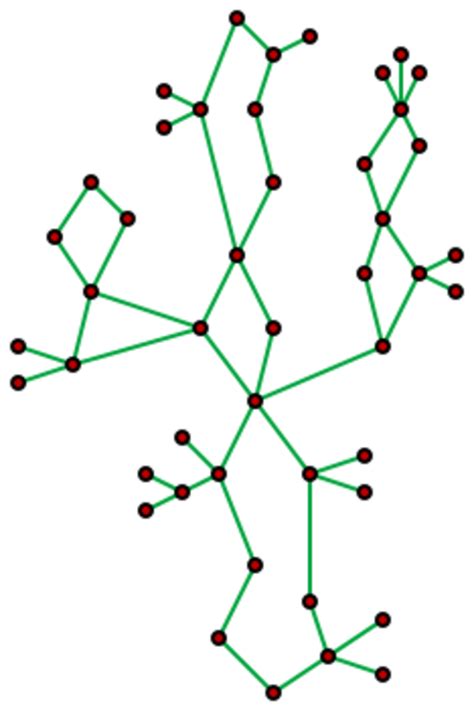}
\caption{Cactus (Eppstein \cite{cactus}).}
\label{cactus}
\end{figure}
 
\begin{prp}
\label{prop:cactus}
Let $G$ be a cactus containing cycles $C^i_{n_i}, i=1, \ldots , m$
for $m \ge 1$. Then
\[
\Vol(\mathrm{Cut}(G))
=
\prod_{i=1}^m \Vol(\mathrm{Cut}(C^i_{n_i})).
\]
\end{prp}
\begin{proof}
A cactus can be decomposed by repeatedly deleting
a pendant vertex or cycle until a single cycle $C$ remains. Reversing
this process, we can reconstruct the cactus, updating the volume
at each step. The initial volume is $\Vol(\mathrm{Cut}(C))$.
Repeatedly using Proposition \ref{union}, appending
an edge does not change the volume and appending a cycle $C'$ results in
the volume being multiplied by $\Vol(\mathrm{Cut}(C'))$. The proposition follows.
\end{proof}

Combining with Proposition \ref{cycle}, we can readily compute
the volume of $\mathrm{Cut}(G)$ for cacti.
To illustrate, the volume of the cactus given in
Figure \ref{cactus} is

\begin{equation*}
\Vol(\mathrm{Cut}(\text{cactus}))
=
\Vol(\mathrm{Cut}(C_8))~
\Vol(\mathrm{Cut}(C_6))~
\Vol(\mathrm{Cut}(C_4))^4~
\Vol(\mathrm{Cut}(C_3))
=\frac{314}{315}\cdot\frac{43}{45}\cdot
\left(\frac{2}{3}\right)^4\cdot\frac{1}{3}.
\end{equation*}
 
Next, we consider necklaces which are formed by taking a cycle
$C_n$ and appending a cycle  $C^i_{n_i}, i=1, \ldots , n$
at each of its vertices. An example is given in Figure \ref{neck}.
\begin{figure} [H]
\centering
\includegraphics[width=5cm]{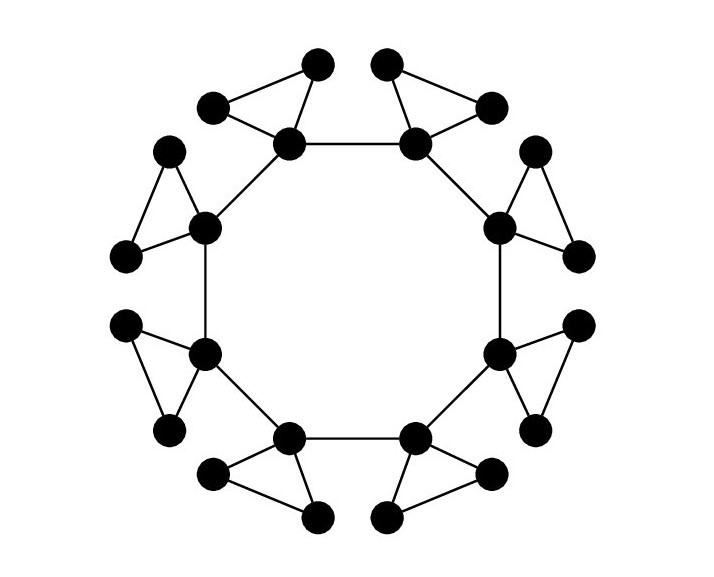}
\caption{8-necklace: $\text{neck}_8$ (Lee and Skipper \cite{LS20}).}
\label{neck}
\end{figure}

\begin{prp}
\label{prop:necklace}
Let $G$ be a necklace based on $C_n$.
\[
\Vol(\mathrm{Cut}(G))
=
\Vol(\mathrm{Cut}(C_n))
\prod_{i=1}^n \Vol(\mathrm{Cut}(C^i)).
\]
\end{prp}
\begin{proof}
The proof is similar to Proposition \ref{prop:cactus}.
$G$ can be decomposed by deleting the cycles $C^i, i=1,...,n$
resulting in $C_n$. Reversing the process and updating the
volume, the proposition follows.
\end{proof}

To illustrate, the volume of the necklace given in
Figure \ref{neck} is
\[
\Vol(\mathrm{Cut}(\text{neck}_8))
=
\Vol(\mathrm{Cut}(C_8))~
\Vol(\mathrm{Cut}(C_3))^8=\frac{314}{315}\cdot\frac{1}{3^8}.
\]

Lee and Skipper \cite{LS20} calculated the volumes of the correlation polytope of
the star, path and cycle.
They index stars and paths by their number of edges. Accordingly, their
$n$-edge star and path are denoted here by $S_{n+1}$ and $P_{n+1}$,
respectively. For cycles no shift is needed.
The covariance map gives the volume of the cut polytope of the
corresponding suspensions.
Their results, along with those earlier in this section,
are given in Table \ref{LS}, where $F_n$ denotes a forest on $n$ nodes that has at least one edge.

\begin{table} [H]
\caption{Volumes of sparse cut polytopes.}
\label{LS}
  \begin{center}
\renewcommand{\arraystretch}{1.2}
    \begin{tabular}{|c|c|c|}
\hline
$G$ & $\Vol(\CG)$ & $\Vol(\mathrm{Cut}(\nabla G))$    \\
  & (this paper) & Lee and Skipper \cite{LS20} \\
\hline
$S_{n+1}$ & 1 & $ \frac{2^n(n!)^2}{(2n+1)!}$ \normalsize  \\
\hline
$P_{n+1}$ & 1 & $\frac{A_{2n+1}}{(2n+1)!}$ \\
\hline
$F_n$ & 1 &  open  \\
\hline
$C_n$ &  $1 -\frac{2^{n-1}}{n!}$ &  $\frac{nA_{2n-1}-2^{2n-2}}{(2n)!}$ \\
\hline
cactus &  $\prod_{i=1}^m \Vol(\mathrm{Cut}(C^i_{n_i}))$ \normalsize & open\\
\hline
necklace  &  $ \Vol(\mathrm{Cut}(C_n))
\prod_{i=1}^n \Vol(\mathrm{Cut}(C^i))$ \normalsize & open \\
\hline
    \end{tabular}
  \end{center}
\begin{center}
$A_n = $ number of alternating permutations of $[n]$,\\ also known as 
Andr\'e or Euler numbers \cite{andre1881memoire, oeisA000111}.
\end{center}
\end{table}

We now consider the volume of the elliptope
for the graphs described in this section.

\begin{prp}
\label{ellunion}
If graphs $G_1$ and $G_2$ are disjoint or share one vertex, then
\[
\mathcal E(G_1\cup G_2)
=
\mathcal E(G_1)\times\mathcal E(G_2)
\]
and
\[
\mathcal I(G_1\cup G_2)
=
\mathcal I(G_1)\times\mathcal I(G_2).
\]
\end{prp}

\begin{proof}
The result is immediate for disjoint graphs by taking orthogonal
Gram representations. Suppose $G_1$ and $G_2$ share a vertex $v$.
Choose Gram representations in which the vector corresponding to $v$
is the same unit vector $u$. The components orthogonal to $u$ in the
two representations may be placed in mutually orthogonal subspaces.
This preserves all inner products corresponding to edges of
$G_1$ and $G_2$. The reverse inclusion follows by restriction.
The result for $\mathcal I$ follows from (\ref{E2I}).
\end{proof}

For graphs with no $K_4$-minor (series-parallel graphs),
we have the following result of
Laurent \cite{DL97} (Theorem 31.3.7).
\begin{thm}
\label{Laurent}
The following are equivalent for a graph $G$:
\begin{eqnarray*}
&&\EG = \{ x=\cos( \pi~a )~| a~\in~\CG \}. \\
&&\EG = \{ x=\cos( \pi~a )~| a~\in~\MG \}. \\
&&G \text{ has no } K_4 \text{ minor}.
\end{eqnarray*}
\end{thm}
Equivalently, $\Phi_G$ is a bijection from $\IG$ onto $\CG$
if and only if $G$ has no $K_4$ minor.

Applying this to $G=K_2$ we see that $\mathcal{E}(K_2)$ is the
line segment $-1 \le x \le 1$ and hence $\mathcal{I}(K_2)=\mathrm{Cut}(K_2)$,
the line segment $0 \le x \le 1$.
Therefore $\Vol(\mathcal{I}(K_2))=1$. Using Proposition
\ref{ellunion} and following the argument of Proposition \ref{for}
we conclude that for any forest $F$, $\Vol(\mathcal{I}(F))=1$.

We now give an exact formula for the volume of the elliptope of every cycle.

\begin{thm}
\label{cycleelliptope}
\[
\Vol(\mathcal E(C_n))
=
2^{n-1}\sum_{m=1}^{n-2}a_m,
\]
or equivalently,
\[
\Vol(\mathcal I(C_n))
=
\frac12\sum_{m=1}^{n-2}a_m,
\]
where
\[
a_1=\frac{\pi^2}{8},
\qquad
a_2=\frac{\pi^2}{16},
\]
and, for $m\ge3$, 
\[
a_m
=
\frac{2m-1}{2m+2}a_{m-1}
-
\frac{\pi^2}{4m(m+1)}a_{m-2}.
\]
For example,
\[
\Vol(\mathcal E(C_3))=\frac{\pi^2}{2},\qquad
\Vol(\mathcal E(C_4))=\frac{3\pi^2}{2},\qquad
\Vol(\mathcal E(C_5))
=\frac{\pi^2(87-\pi^2)}{24}.
\]
\end{thm}

\begin{proof}
Since $C_n$ has no $K_4$ minor, Theorem \ref{Laurent} implies that
$\Phi_{C_n}$ is a bijection from $\mathcal I(C_n)$ onto $\mathrm{Cut}(C_n)$.
Identify the edges of $C_n$ with $[n]$. Writing $z=\Phi_{C_n}(x)$, we have
\[
x_i
=
\frac{1-\cos(\pi z_i)}2,
\qquad i=1,\ldots,n.
\]
By Theorem \ref{seymour},
$\mathrm{Cut}(C_n)$ is described by the box inequalities $0\le z_i\le1$ together with,
for every odd set $F\subseteq[n]$, the cycle inequality
\[
\sum_{i\in F}z_i-\sum_{i\notin F}z_i\le |F|-1.
\]
For each odd $F\subseteq[n]$, let $R_F$ denote the coordinate reflection
\[
(R_Fz)_i=
\begin{cases}
1-z_i,&i\in F,\\
z_i,&i\notin F.
\end{cases}
\]
Writing $w=R_Fz$, the closure of the region violating this inequality is mapped onto the standard simplex
\[
\Delta_n=\left\{w\in\mathbb R_{\ge0}^n:\sum_{i=1}^n w_i\le1\right\}.
\]
Thus each odd $F$ cuts off a reflected simplex from the unit cube.
There are $2^{n-1}$ such simplices, and their interiors are disjoint.
Since $\Phi_{C_n}^{-1}$ acts coordinatewise, its Jacobian determinant is
\[
\prod_{i=1}^n \frac{\pi}{2}\sin(\pi z_i)
=
\left(\frac\pi2\right)^n
\prod_{i=1}^n\sin(\pi z_i).
\]
Since the sine factors are invariant under $R_F$, this Jacobian is unchanged.
Applying the change-of-variables formula to $\Phi_{C_n}^{-1}$ gives
\begin{equation}
\label{eq:volumeI}
\Vol(\mathcal I(C_n))
=
\int_{\mathrm{Cut}(C_n)}
\left(\frac\pi2\right)^n
\prod_{i=1}^n\sin(\pi z_i)\,dz.
\end{equation}
The integral over $[0,1]^n$ is
\[
\prod_{i=1}^n
\int_0^1\frac\pi2\sin(\pi z_i)\,dz_i
=1.
\]
For $t\ge0$, define
\begin{equation}
\label{eq:Jn}
J_n(t)
:=
\int_{\substack{u_i\ge0\\u_1+\cdots+u_n\le t}}
\prod_{i=1}^n\sin u_i\,du.
\end{equation}
For $t=\pi$, the scaling $u_i=\pi w_i$ maps $\Delta_n$ to the domain of $J_n(\pi)$. Hence the contribution of any one removed simplex is
\[
2^{-n}J_n(\pi).
\]
Since there are $2^{n-1}$ removed simplices, \eqref{eq:volumeI} becomes
\begin{equation}
\label{eq:volI-J}
\Vol(\mathcal I(C_n))
=
1-\frac12J_n(\pi).
\end{equation}
By \eqref{VE2I},
\begin{equation}
\label{eq:volE-J}
\Vol(\mathcal E(C_n))
=
2^n-2^{n-1}J_n(\pi).
\end{equation}
Define
\[
g_1(t)=\sin t,
\qquad
g_m(t)=\int_0^t g_{m-1}(u)\sin(t-u)\,du,
\qquad m\ge2.
\]
With this notation, \eqref{eq:Jn} can be rewritten as
\[
J_n(t)=\int_0^t g_n(v)\,dv,
\]
and hence
\begin{equation}
\label{eq:Jprime}
J_n'(t)=g_n(t),
\qquad
J_n(0)=0.
\end{equation}
For $m\ge2$, differentiating the defining relation for $g_m$ gives
\[
g_m'(t)=\int_0^t g_{m-1}(u)\cos(t-u)\,du
\]
and
\begin{equation}
\label{eq:id2}
g_m''(t)+g_m(t)=g_{m-1}(t).
\end{equation}
Summing \eqref{eq:id2} over $m=2,\ldots,n$ and using $g_1''=-g_1$ gives
\[
\sum_{m=1}^n g_m''(t)=-g_n(t).
\]
Integrating from $0$ to $t$ and using \eqref{eq:Jprime}, $g_1'(0)=1$, and $g_m'(0)=0$ for $m\ge2$ gives
\begin{equation}
\label{eq:Jsum}
J_n(t)
=
1-\sum_{m=1}^n g_m'(t).
\end{equation}
Since $g_1'(\pi)=-1$ and, from the preceding formula for $g_m'$, $g_2'(\pi)=0$, \eqref{eq:Jsum} gives
\begin{equation}
\label{eq:Jpi}
J_n(\pi)
=
2-\sum_{m=3}^n g_m'(\pi).
\end{equation}
It remains to obtain a recurrence for $g_m'(\pi)$. Since $|\sin t|\le1$,
\[
|g_m(t)|\le\frac{t^{m-1}}{(m-1)!},
\qquad
|J_n(t)|\le\frac{t^n}{n!}.
\]
Hence the Laplace transforms exist for $s>0$. Write
\[
\mathcal L h(s)=\int_0^\infty e^{-st}h(t)\,dt,\qquad s>0.
\]
We use the standard identities
\[
\mathcal L(h')(s)=s\mathcal L h(s)-h(0),
\qquad
\mathcal L(f*g)(s)=\mathcal L f(s)\mathcal L g(s),
\]
where
\[
(f*g)(t)=\int_0^t f(u)g(t-u)\,du.
\]
Since $g_m$ is the $m$-fold convolution of $\sin t$ with itself and
\[
\mathcal L(\sin t)(s)=\frac1{s^2+1},
\]
the convolution theorem gives
\[
\mathcal L g_m(s)=\frac1{(s^2+1)^m}.
\]
Since $g_{m+1}(0)=0$,
\[
\mathcal L(g_{m+1}')(s)=\frac{s}{(s^2+1)^{m+1}}=-\frac1{2m}\frac{d}{ds}(s^2+1)^{-m}.
\]
Using
\[
-\frac{d}{ds}\mathcal L h(s)=\mathcal L\!\left(t\,h(t)\right)(s),
\]
and uniqueness of the Laplace transform, we obtain
\begin{equation}
\label{eq:id1}
g_{m+1}'(t)
=
\frac{t}{2m}g_m(t),
\qquad m\ge1.
\end{equation}
At $t=\pi$, \eqref{eq:id1} gives
\[
g_m(\pi)=\frac{2m}{\pi}g_{m+1}'(\pi),
\qquad
g_{m+1}(\pi)=\frac{2(m+1)}{\pi}g_{m+2}'(\pi).
\]
Differentiating \eqref{eq:id1} and using \eqref{eq:id2} gives
\[
\begin{aligned}
\frac{2(m+1)}{\pi}g_{m+2}'(\pi)
&=g_{m+1}(\pi)\\
&=g_m(\pi)-g_{m+1}''(\pi)\\
&=\frac{2m-1}{\pi}g_{m+1}'(\pi)
-\frac{\pi}{2m}g_m'(\pi).
\end{aligned}
\]
Hence, for $m\ge1$,
\begin{equation}
\label{eq:g-rec}
g_{m+2}'(\pi)
=
\frac{2m-1}{2m+2}g_{m+1}'(\pi)
-
\frac{\pi^2}{4m(m+1)}g_m'(\pi).
\end{equation}
Set
\[
a_m:=g_{m+2}'(\pi),
\qquad m\ge1.
\]
Taking $m=1,2$ in \eqref{eq:g-rec} gives the stated initial values, and for $m\ge3$ it gives the recurrence in the theorem.
Since $a_m=g_{m+2}'(\pi)$, reindexing \eqref{eq:Jpi} gives
\[
J_n(\pi)
=
2-\sum_{m=1}^{n-2}a_m.
\]
Substitution into \eqref{eq:volE-J} and \eqref{eq:volI-J} gives the asserted formulas.
\end{proof}

For large $n$, $\mathcal I(C_n)$ provides a very tight volume approximation of $\mathrm{Cut}(C_n)$, as evidenced by the following corollary.

\begin{crl}
For $n\ge4$,
\[
1\le
\frac{\Vol(\mathcal I(C_n))}{\Vol(\mathrm{Cut}(C_n))}
\le
1+\frac{2^{n-1}}{n!-2^{n-1}}
\le
\frac{\Vol(\mathcal I(C_3))}{\Vol(\mathrm{Cut}(C_3))}
=
\frac{3\pi^2}{16}
\approx 1.8506.
\]
Moreover, the upper bound is strictly decreasing in $n$ and converges to $1$.
\end{crl}

\begin{proof}
The lower bound follows from $\mathrm{Cut}(C_n)\subseteq\mathcal I(C_n)$.
For $n=3$, Theorem \ref{cycleelliptope} and Proposition \ref{cycle} give
\[
\frac{\Vol(\mathcal I(C_3))}{\Vol(\mathrm{Cut}(C_3))}
=
\frac{3\pi^2}{16}.
\]
For $n\ge4$, $\mathcal I(C_n)\subseteq[0,1]^n$ and Proposition \ref{cycle} give
\[
\frac{\Vol(\mathcal I(C_n))}{\Vol(\mathrm{Cut}(C_n))}
\le
1+\frac{2^{n-1}}{n!-2^{n-1}}
<
\frac{3\pi^2}{16}.
\]
One easily checks that the explicit upper bound decreases monotonically to $1$.
\end{proof}
Together with the product formulas above, Corollary 3.10 gives analogous volume-ratio bounds for cacti and necklaces.


\section{Asymptotic volume of the elliptope}
\label{asymptotic}

The volume of $\En$ was given recursively by Joe \cite{Joe2006} as follows:
$$
\Vol(\mathcal{E}_2) = 2,
\Vol(\En) = \Vol(\mathcal{E}_{n-1}) 2^{(n-1)^2} ( B(n/2, n/2))^{n-1}, n \ge 3,
$$
where $B(a,b) = \Gamma (a) \Gamma (b) / \Gamma (a+b)$,
$a,b > 0$, is the beta function, and $\Gamma$ denotes Euler's gamma function.

The objective of this section is to derive precise asymptotics for $\Vol(\En)$ as $n \to \infty$.

\begin{thm}
    {\label{logvolume} }
There exist finite constants $a$ and $b$ such that
$$
a \le \log (\Vol(\En)) - \log (v_n) \le b,
$$
where
\begin{align*}
&\log (v_n) =\\
&-\frac {n^2 \log (n)}{4} 
+ n^2 \left( \frac{1 } {8} + \frac{\log (\sqrt{2 \pi} )}{2 }\right)
+ \frac{ n \log ( n) }{4}
- n \left(\frac{1}{4} + \frac{\log (\sqrt{2 \pi} )}{2 }\right)
- \frac {\log (n)}{24}.\\
\end{align*}
\end{thm}

\begin{proof}
$$
\log (\Vol(\En)) = \log (2) + \sum_{j=2}^{n-1} \left(
      j^2 \log(2) + 2j \log ( \Gamma ((j+1)/2 ) )
      - j \log ( \Gamma (j+1) )
\right).
$$
From Stirling's series (see, e.g., Whittaker and Watson \cite{WW1963} and Alzer \cite{Alzer1997}), we retain that
for a fixed integer $N \ge 1$, 
$$
\log (\Gamma (x)) = (x-1/2) \log (x) - x + (1/2) \log (2 \pi ) + \sum_{k=1}^{N-1} \frac { B_{2k} } { 2k (2k-1) x^{2k-1}} + R_N,
$$
where
$$
R_N = \theta \frac { B_{2N} } { 2N (2N-1) x^{2N-1}},
$$
for some $\theta \in [0,1]$, where different occurrences of
$\theta$ may denote different values,
and 
$B_{2k}$ is the $2k$-th Bernoulli number. For our computations, we recall that $B_0 = 1$, $B_2 = 1/6$, $B_4 = -1/30$, and $B_6 = 1/42$. 
For example,
$$
\log (\Gamma (x)) = (x-1/2) \log (x) - x + (1/2) \log (2 \pi ) + \frac { 1 } { 12 x} - \frac { 1 } { 360 x^3}
+ \frac { \theta } { 1260 x^5}.
$$
Using this representation and $\Gamma (x+1) = x \Gamma (x)$, we obtain
\begin{align*}
&\log (\Vol(\En)) \\
&= \log (2) + 
\sum_{j=2}^{n-1} 
\left(
      j^2 \log(2) 
      + 
       j^2 \log \left(\frac{j+1}{2}\right) - j(j+1) + j \log (2 \pi ) \right)\\
     &+ 
\sum_{j=2}^{n-1} 
\left(  
        \frac { j } { 3 (j+1)} - \frac { 2j } { 45 (j+1)^3}
+ \frac { \theta 16 j } { 315 (j+1)^5}\right) \\
&+ 
\sum_{j=2}^{n-1} 
\left(
      - j (j+1/2) \log (j) + j^2 - \frac{j}{2} \log (2 \pi ) - \frac { 1 } { 12 } + \frac { 1 } { 360 j^2}
- \frac { \theta } { 1260 j^4}
\right) \\
&= \log (2) + 
\sum_{j=2}^{n-1} 
\left(
       j^2 \log (1+1/j) - j + \frac{j}{2} \log (2 \pi ) + \frac { j } { 3 (j+1)}\right)\\
       &+ 
\sum_{j=2}^{n-1} 
\left(
       - \frac { 2j } { 45 (j+1)^3} 
+ \frac { \theta 16  } { 315 (j+1)^4}
- \frac{j}{2} \log (j)   - \frac { 1 } { 12 } + \frac { 1 } { 360 j^2}
- \frac { \theta } { 1260 j^4}
\right)\\
&= \log (2) + 
\sum_{j=2}^{n-1} 
\left(
       j^2 \log (1+1/j) - j + \frac{j}{2} \log (2 \pi ) + \frac { 1 } { 4 } \right)\\
       &+ 
\sum_{j=2}^{n-1} 
\left( - \frac { 1 } { 3 (j+1)}- \frac { 2} { 45 (j+1)^2} + \frac { 2} { 45 (j+1)^3}
+ \frac { \theta 16  } { 315 (j+1)^4}
- \frac{j}{2} \log (j)    + \frac { 1 } { 360 j^2}
- \frac { \theta } { 1260 j^4}
\right).
\end{align*}
The last expression is not more than
\begin{align*}
    &\log 2 + \sum_{j=2}^{n-1}
          \left( - \frac{j}{2} \log (j) + j^2 \log (1+1/j) - j + \frac{j}{2} \log (2 \pi ) + \frac { 1 } { 4 } - \frac { 1 } { 3 (j+1)} \right) \\
          &
          - \frac { 2} { 45 } (\zeta(2)-5/4) + O(1/n)
          + \frac { 2} { 45 } ( \zeta(3)- 9/8) 
          + \frac { 16  } { 315 } (\zeta(4) - 17/16)
          + \frac { 1 } { 360 } (\zeta(2)-1) \\
          &\le \log 2 
          - \frac{1}{2} \sum_{j=2}^{n-1} j \log (j) 
          + \sum_{j=2}^{n-1} ( j^2 \log (1+1/j)  - j ) \\
          &
          + \frac {n(n-1) -2} {2}  \log (\sqrt{2 \pi} ) 
          + \frac { n-2 } { 4 } 
          - \frac { H_n - 3/2 } { 3 }
          - \frac { 2(\zeta(2)-5/4)} { 45 }  + O(1/n) \\
          &
          + \frac { 2( \zeta(3)- 9/8) } { 45 } 
          + \frac { 16 (\zeta(4) - 17/16)  } { 315 } 
          + \frac { \zeta(2)-1 } { 360 } ,\\
\end{align*}
where $H_n=\sum_{k=1}^n 1/k$ is the $n$-th harmonic number and
$\zeta(s) = \sum_{k=1}^\infty \frac{1}{k^s}$ is Riemann's zeta function.
Similarly, $\log (\Vol(\En))$ is at least
\begin{align*}
&\log (2) + 
\sum_{j=2}^{n-1} 
\left(
       j^2 \log (1+1/j) - j 
       + \frac{j}{2} \log (2 \pi ) + \frac { 1 } { 4 } 
       - \frac { 1 } { 3 (j+1)}
       - \frac { 2} { 45 (j+1)^2} 
       + \frac { 2} { 45 (j+1)^3} \right)
       \\
       &+ 
\sum_{j=2}^{n-1} 
\left(
        - \frac{j}{2} \log (j)
        + \frac { 1 } { 360 j^2}
        - \frac { 1 } { 1260 j^4}
\right)\\
&\ge\log (2) 
-\sum_{j=2}^{n-1} \frac{j}{2} \log (j)
+ \sum_{j=2}^{n-1} \left( j^2 \log (1+1/j) - j \right) 
+ \frac {n(n-1) -2} {2}  \log (\sqrt{2 \pi })
\\
&\qquad+ \frac { n-2 } { 4 } 
- \frac { H_n - 3/2 } { 3 }
       - \frac { 2 (\zeta(2)-5/4)} { 45 } 
       + \frac { 2(\zeta(3)-9/8)} { 45 }\\
       &\qquad
        + \frac { \zeta(2)-1 } { 360 }
        - \frac { \zeta(4)-1 } { 1260 }
        - O(1/n^2).
\\
\end{align*}
The difference between the upper and lower bounds is a constant plus $ O (1/n)$. We are left with the two summation signs in the bounds.
By Taylor's series with remainder, for some $\theta \in (0,1)$,
we note that
\begin{align*}
j^2 \log (1+1/j) - j
&= j^2 \left(
\frac{1}{j}-\frac{1}{2j^2}+\frac{1}{3j^3}
-\frac{1}{4(j+\theta)^4}\right)-j\\
&= -\frac{1}{2}+\frac{1}{3j}
-\frac{j^2}{4(j+\theta)^4}.
\end{align*}
Therefore,
$$
\sum_{j=2}^{n-1} (j^2 \log (1+1/j) - j) 
$$
is not more than
$$
-\frac{n-2}{2}
+ \frac{H_{n-1}-1}{3},
$$
and is at least
$$
-\frac{n-2}{2}
+ \frac{H_{n-1}-1}{3}
- \frac{\zeta(2)-1}{4}.
$$
The Barnes G-function 
$$
G(n)  
= \prod_{i=1}^{n-1} \Gamma (i)
= \prod_{i=1}^{n-2} i!
= \prod_{i=1}^{n-2} i^{n-1-i}
$$
will be helpful.  In particular,
$$
\log (G(n+1))
= \sum_{j=1}^{n-1} (n-j) \log (j)
= n \log (\Gamma (n)) - \sum_{j=1}^{n-1} j \log (j).
$$
We recall an asymptotic expansion of the Barnes G-function \cite{Nemes,Ferreira}. 
$$
\log (G(n+1))
= 
\frac{n^2}{4} 
+ n \log ( \Gamma (n+1)) 
- \frac {n(n+1) \log (n)}{2}
- \frac {\log (n)}{12}
- \log (A)
+ O \left(\frac{1}{n^2}\right),
$$
where the constant $A$ is the Glaisher--Kinkelin constant, and can be expressed as 
$
\log (A) = 
\frac{1}{12} 
- \zeta'(-1)
= 0.2487\ldots.
$
Thus, 
$$
\sum_{j=1}^{n-1} j \log (j)
= 
\frac {n(n+1) \log (n)}{2}
-\frac{n^2}{4} 
- n \log ( n) 
+ \frac {\log (n)}{12}
+ \log (A)
+ O \left(\frac{1}{n^2}\right).
$$
Collecting all this shows that $\log (\Vol(\En))$ differs by a bounded term from
$$
-\frac {n^2 \log (n)}{4} 
+ n^2 \left( \frac{1 } {8} + \frac{\log (\sqrt{2 \pi} )}{2 }\right)
+ \frac{ n \log ( n) }{4}
- n \left(\frac{1}{4} + \frac{\log (\sqrt{2 \pi} )}{2 }\right)
- \frac {\log (n)}{24}.
$$
\end{proof}

\begin{crl}
\label{cor:IR}
\[
\lim_{n\to\infty}
\left(
\frac{1}{n^2}\log\frac{\Vol(\In)}{\Vol(\Rn)}
+\frac14\log n
\right)
=
\frac18+\frac14\log\frac{\pi}{2}.
\]
\end{crl}

\begin{proof}
The result follows from Theorem \ref{logvolume}, (\ref{VE2I}), Theorem \ref{KLS}, and Stirling's approximation.
\end{proof}

\section{Asymptotic volumes of the metric and cut polytopes}
\label{asymptotic2}

For general graphs, we have the following result of
Deza and Laurent \cite{DL97} (Theorem 31.2.1).

\begin{thm}
\label{DLangles}
Let $v_1,\ldots,v_n$ be $n$ unit vectors in $\mathbb R^m$
$(n\ge3,\ m\ge1)$.
Let $a=(a_{ij})_{1\le i<j\le n}\in\mathbb R^{\binom n2}$ and $a_0\in\mathbb R$ such that the
inequality $a^T x\le a_0$ is valid for the cut polytope $\mathrm{Cut}_n$. Then
\[
\sum_{1\le i<j\le n}
a_{ij}\arccos(v_i^Tv_j)
\le \pi a_0.
\]
\end{thm}

\begin{prp}
\label{angularinclusion}
$
\Phi_{K_n}(\In)\subseteq\Cn.
$
\end{prp}
\begin{proof}

Let $x\in\In$. By the definition of $\In$, there is a symmetric
positive semidefinite matrix $Y=[y_{ij}]$ such that
$
y_{ii}=1,\qquad y_{ij}=1-2x_{ij}
$
for $1\le i<j\le n$. Hence $Y$ is the Gram matrix of unit vectors
$v_1,\ldots,v_n$, and therefore
$
1-2x_{ij}=v_i^Tv_j
$
for $1\le i<j\le n$.
Let $a^Tz\le a_0$ be any valid inequality for $\Cn$.
By Theorem \ref{DLangles},
\[
\sum_{i<j}a_{ij}\arccos(v_i^Tv_j)\le\pi a_0.
\]
Since
$
\Phi_{K_n}(x)_{ij}
=
\frac1\pi\arccos(v_i^Tv_j),
$
we obtain
$
a^T\Phi_{K_n}(x)\le a_0.
$
Thus $\Phi_{K_n}(x)$ satisfies every valid inequality for $\Cn$,
and hence $\Phi_{K_n}(x)\in\Cn$.
\end{proof}

\begin{thm}
\label{cutelliptope}
\[
\left(\frac{2}{\pi}\right)^{\binom{n}{2}}\Vol(\In)
\le
\Vol(\Cn)
\le
\Vol(\In).
\]
In particular,
\[
\Vol(\Cn)=\exp\left(-\frac{n^2\log n}{4}+O(n^2)\right).
\]
\end{thm}

\begin{proof}
Proposition \ref{angularinclusion} implies that 
$
\Vol(\Phi_{K_n}(\In))
\le
\Vol(\Cn).
$
Since $\Phi_{K_n}$ acts coordinatewise and
\[
\frac{d}{dt}
\left(
\frac{1}{\pi}\arccos(1-2t)
\right)
=
\frac{1}{\pi\sqrt{t(1-t)}},
\]
its Jacobian determinant is
\[
\prod_{i<j}
\frac{1}{\pi\sqrt{x_{ij}(1-x_{ij})}}.
\]
Since $\Phi_{K_n}$ is injective, the change-of-variables formula gives
\[
\Vol(\Phi_{K_n}(\In))
=
\int_{\In}
\prod_{i<j}
\frac{1}{\pi\sqrt{x_{ij}(1-x_{ij})}}\,dx.
\]
Since
\[
\frac{1}{\pi\sqrt{t(1-t)}}\ge\frac{2}{\pi}
\qquad (0<t<1),
\]
we obtain
\[
\Vol(\Phi_{K_n}(\In))
\ge
\left(\frac{2}{\pi}\right)^{\binom n2}
\Vol(\In),
\]
which proves the lower bound. The upper bound follows from $\Cn\subset\In$.
Finally, Theorem \ref{logvolume} and (\ref{VE2I}) give
\[
\Vol(\In)=\exp\left(-\frac{n^2\log n}{4}+O(n^2)\right),
\]
and the bounds above imply the last assertion.
\end{proof}


\begin{thm}
\label{MetRecurrence}
\[
\left(\frac{1}{3}\right)^{\binom{n}{2}}
\le
\Vol(\Mn)
\le
\Vol(\mathrm{Met}_{n-1})^{n/(n-2)}.
\]
\end{thm}

\begin{proof}
The lower bound follows from the observation that $\Mn$ contains the cube
$[1/3,2/3]^{\binom{n}{2}}$, since the $4\binom{n}{3}$ triangle inequalities
are satisfied for $x_{ij},x_{ik},x_{jk}\in[1/3,2/3]$.

The upper bound is immediate for $n=3$. For $n\ge4$, consider, for
$k=1,\dots,n$,
\[
\mathrm{Met}_n^k
=
\left\{
x\in[0,1]^{\binom n2}:
(x_{ij})_{\substack{i<j\\ i,j\ne k}}\in\mathrm{Met}_{n-1}
\right\}.
\]
Up to relabelling of the coordinates,
$\mathrm{Met}_n^k$ is congruent to
$\mathrm{Met}_{n-1}\times[0,1]^{n-1}$, and hence
\[
\Vol(\mathrm{Met}_n^k)=\Vol(\mathrm{Met}_{n-1}).
\]
Moreover,
\[
\Mn=\bigcap_{k=1}^n\mathrm{Met}_n^k.
\]
Let $\mathbf 1_k$ denote the indicator of $\mathrm{Met}_n^k$.
The function $\mathbf 1_k$ depends only on the coordinates $x_{ij}$ with
$i,j\ne k$. Thus each coordinate $x_{ij}$ occurs in precisely the $n-2$
functions $\mathbf 1_k$ with $k\notin\{i,j\}$.
Applying the generalized H\"older inequality for product spaces with weight
$1/(n-2)$ to each of the functions $\mathbf 1_k$ gives
\begin{align*}
\Vol(\Mn)
&=
\int_{[0,1]^{\binom n2}}
\prod_{k=1}^n \mathbf 1_k(x)\,dx\\
&\le
\prod_{k=1}^n
\left(
\int_{[0,1]^{\binom n2}}
\mathbf 1_k(x)^{\,n-2}\,dx
\right)^{1/(n-2)}\\
&=
\prod_{k=1}^n
\left(
\int_{[0,1]^{\binom n2}}
\mathbf 1_k(x)\,dx
\right)^{1/(n-2)}\\
&=
\prod_{k=1}^n
\Vol(\mathrm{Met}_n^k)^{1/(n-2)}\\
&=
\Vol(\mathrm{Met}_{n-1})^{n/(n-2)}.
\end{align*}
\end{proof}

\begin{crl}
\label{cor:MetAsymptotic}
There exists a constant $\lambda_{\mathrm{Met}}$
in $[0.7027,\log 3]$ such that
\[
-\frac{\log\Vol(\Mn)}{\binom{n}{2}}
\longrightarrow
\lambda_{\mathrm{Met}}.
\]
Equivalently,
\[
\log\Vol(\Mn)
=
-\lambda_{\mathrm{Met}}\binom{n}{2}
+o(n^2).
\]
\end{crl}

\begin{proof}
Theorem \ref{MetRecurrence} implies that
\[
-\frac{\log\Vol(\Mn)}{\binom{n}{2}}
\]
is nondecreasing, while the lower bound
$\Vol(\Mn)\ge 3^{-\binom{n}{2}}$
shows that it is bounded above by $\log 3$.
Hence the sequence converges to a constant
$\lambda_{\mathrm{Met}}\le\log 3$.
Using the exact value for $n=7$ gives
\[
\lambda_{\mathrm{Met}}
\ge
-\frac{1}{21}
\log\left(
\frac{3586321206047}{9206753675223604500}
\right)
>0.7027.
\]
\end{proof}

\begin{crl}
\label{cor:IMet}
\[
\frac{\Vol(\In)}{\Vol(\Mn)}
=
\exp\left(
-\frac{n^2\log n}{4}
+\Theta(n^2)
\right),
\]
and in particular
\[
\frac{\Vol(\In)}{\Vol(\Mn)}
\longrightarrow0.
\]
\end{crl}

\begin{proof}
The result follows from Corollary \ref{cor:MetAsymptotic},
Theorem \ref{logvolume}, and (\ref{VE2I}).
\end{proof}


\section{Volume estimates from Volesti}
\label{estimation}
 
In this section we use the software Volesti\footnote{\url{https://volesti.readthedocs.io}}
to get estimates for volumes that are too hard to compute exactly.
In particular, we will get estimates for the volume of
metric polytopes to
compare with the known volumes of the elliptopes.
All runs reported here were performed using the R interface and default Volesti settings.
Exact polytope volume computation is \#P-hard, but Volesti appears to give
good estimates for the volumes of the metric and cut polytopes.

In Table \ref{tab:avolumes} we first compare the estimates with the exact
computed values, as shown in Table \ref{tab:volumes}, converted to floating point.
The estimates are the mean value over 20 runs using
the Cooling Convex Bodies (CB) \cite{CEF23} algorithm in Volesti. For the hardest problem, $\mathrm{Cut}_7$,
the 20 runs completed in about 5 minutes. We remark that the estimates for the cut
polytope use its V-representation as input, whereas for the metric polytope,
the input is its H-representation.
This table suggests that for small $n$ the estimates for the volumes of the cut and metric polytopes are
very close to the exact values. 
\begin{table}[htb]
\caption{Exact volumes vs estimates (*) for small $n$}
\label{tab:avolumes}
\setlength{\extrarowheight}{5pt}
\begin{tabular}{||c||c|c||c|c||c|c||}
\hline
n & $\Vol(\Cn)$ & $\Vol(\Cn^*)$ & $\Vol(\Mn)$& $\Vol(\Mn^*)$& $\frac{\Vol(\Mn)}{\Vol(\Cn)}$&$\frac{\Vol(\Mn^*)}{\Vol(\Cn^*)}$\\
 & & (est.) & & (est.) & & (est.) \\
\hline
3 &  0.33     &   0.33  & 0.33 & 0.33    & 1  & 0.99 \\
4 & 4.44e-02  & 4.38e-02 &  4.44e-02    & 4.44e-02    &  1 & 1.01 \\
5 & 2.25e-03   &  2.36e-03    & 2.35e-03  &  2.40e-03     & 1.04     & 1.01 \\
6 &  4.10e-05  &  4.16e-05 & 4.86e-05 & 4.99e-05  &    1.19& 1.20 \\
7 &  2.52e-07&  2.46e-07 &  3.89e-07 & 3.90e-07 &      1.55& 1.58 \\
\hline
\end{tabular}
\end{table}

For larger values of $n$, we cannot compute the exact volume for $\Mn$, so
we rely on the estimates provided by Volesti. The results and
statistics provided by Volesti are given in the Appendix. 
The mean estimates place the crossover near $n=13$.
Again, we
used the CB algorithm, but obtained similar results using the other
two supplied algorithms, Cooling Gaussian (CG) and Sequence of Balls (SOB).
Since the dimension of the convex bodies is $d=\binom{n}{2}$, it seems reasonable to
fit a quadratic curve to the logarithm of the mean estimate
of $\Vol(\Mn)$.
Indeed, the fit is excellent, as shown in Figure \ref{IMn}, which was generated using the model
function in R.
The estimated parabola is $y=-0.50n^2 + 1.62n - 1.66$, or in terms
of the dimension $d$, $y=-d + O(\sqrt(d))$.
It is shown in dashed blue in the figure; solid
blue shows the log of the computed values from
Table \ref{tab:volumes}.
The asymptotic comparison with $\In$ is given analytically
by Corollary \ref{cor:IMet}; the estimates are useful for estimating
$\lambda_{\mathrm{Met}}$ and locating the crossover.
For comparison, Figure \ref{IMn} also
contains log plots of $\Rn$ and $\In$ based
on the exact formulae given earlier.
The parabola is also a near-perfect fit
to the data for $3 \le n \le 7$, for which we have exact results.
The consecutive sample means at $n=20$ and $n=21$ violate
the recurrence bound in Theorem \ref{MetRecurrence} for $n=21$.
The $n=20$ estimates have unusually large dispersion, and values
within the two observed ranges remain compatible with the exact inequality.

\begin{figure} [H]
\centering
\includegraphics[width=12cm]{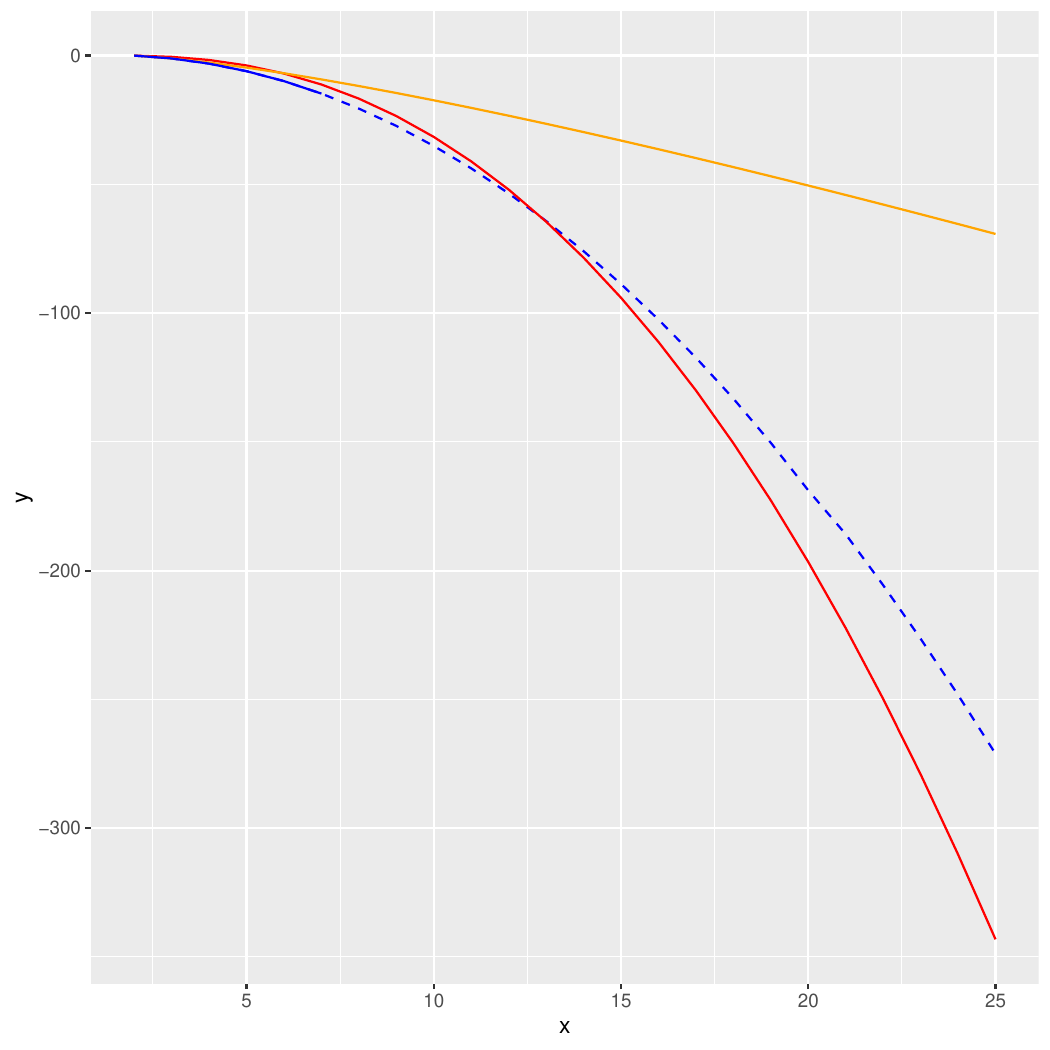}
\caption{Log plots of \textcolor{orange}{$\Vol(\Rn)$}, \textcolor{blue}{$\Vol(\Mn)$ (estimate dashed)} and $\textcolor{red}{\Vol(\In)}$} 
\label{IMn}
\end{figure}

 
\section{Conclusion and open problems}

In this paper, we use volume bounds to
give a measure of how well metric polyhedra approximate
the cut polytope compared to that provided by the elliptope.

First, in Section \ref{exact}, we derived exact formulae for
$\Vol(\CG)$ for various sparse graphs.
Since these graphs have no
$K_5$ minor this is also the same as $\Vol(\MG)$.
Theorem \ref{cycleelliptope} gives an exact computation of the volume of the elliptope of every cycle.
For cycles, the resulting volume bound is extremely tight: the volume ratio of $\mathcal I(C_n)$ to
$\mathrm{Cut}(C_n)$ converges rapidly to one. By the product formulas, similar conclusions hold for the
related cacti and necklaces considered above. In contrast, Table \ref{tab:volumes}
shows that for complete graphs of small order the elliptope gives a relatively weak volume approximation.

In Section \ref{asymptotic}, we gave a tight bound
on the volume of the elliptope of the complete graph $K_n$.
Exact computation for small $n$ showed that the metric polytope
gives a much closer approximation to the cut polytope than the
elliptope does. The rooted metric polytope is already weaker than
the elliptope for $n=6,7$, and asymptotically it is much weaker.
In Section \ref{asymptotic2} we proved that
$
\Vol(\Mn)=\exp(-\Theta(n^2))
$
and
$
\frac{\Vol(\In)}{\Vol(\Mn)}\longrightarrow0.
$
Thus for large $n$ the elliptope gives a much tighter volume
relaxation than the metric polytope. The volume estimates place
the crossover near $n=13$.

\begin{problem}
Determine 
$$
\lambda_{\mathrm{Met}} = \lim_{n\to\infty}
-\frac{\log\Vol(\Mn)}{\binom{n}{2}}.
$$
\end{problem}

\begin{problem}
Determine whether 
$$
\lim_{n\to\infty}
\frac{\log\Vol(\In)-\log\Vol(\Cn)}{\binom{n}{2}}
$$
exists and, if so, determine its value which, by Theorem \ref{cutelliptope}, would lie in
$[0,\log \frac{\pi}{2}]$.
\end{problem}

Finally, we recall Lee and Skipper \cite{LS20} posed the challenge
of finding a polynomial time algorithm for finding $\Vol(\text{Cor}(G))$
for graphs $G$ with no $K_4$ minor (series-parallel graphs).
Via the covariance, this is equivalent to computing $\Vol(\mathrm{Cut}(\nabla G))$.
These are a subset of graphs with no $K_5$ minor and
so $\mathrm{Cut}(\nabla G) = \mathrm{Met}(\nabla G)$. An even bigger
challenge would be to find an algorithm for computing  $\Vol(\CG)$
for any graph $G$ with no $K_5$ minor or show that it is \#P-hard.

\newpage
\appendix
\section{Results of volume estimation for $\Vol(\Mn)$}
 \begin{table}[H]
\caption{Estimates for $\Vol(\Mn)$ vs truncated exact values of $\Vol(\In)$.}
\label{tab:mvolumes}
\setlength{\extrarowheight}{5pt}
\begin{tabular}{||c||c|c|c|c|c|c||c|c||}
\hline
 n & Min.  & 1st Qu. &   Median  &  Mean & 3rd Qu. &   Max.& $\Vol(\In)$ & $\frac{\Vol(\In)}{\mathrm{Mean}(\Mn)}$ \\
\hline
8 &9.86e-10&1.05e-09&1.13e-09&1.15e-09&1.18e-09&1.56e-09  & 5.54e-08 & 47.8 \\
9&9.58e-13&1.31e-12&1.43e-12&1.40e-12&1.52e-12&1.62e-12 & 6.41e-11 & 45.9 \\
10&5.10e-16&5.42e-16&5.77e-16&5.75e-16&5.97e-16&6.45e-16 & 1.93e-14 & 33.7 \\
11&8.08e-20&9.04e-20&9.48e-20&9.59e-20&9.98e-20&1.20e-19 & 1.44e-18 & 15.1 \\
12&4.91e-24&5.46e-24&5.82e-24&5.83e-24&6.15e-24&7.12e-24 & 2.52e-23 & 4.32 \\
13&7.10e-29&1.23e-28&1.31e-28&1.29e-28&1.38e-28&1.65e-28 & 9.89e-29 & 0.76 \\
14&8.53e-34&9.89e-34&1.05e-33&1.04e-33&1.10e-33&1.20e-33 & 8.34e-35 & 0.08\\
15&2.12e-39&2.66e-39&2.93e-39&2.99e-39&3.48e-39&3.90e-39 & 1.45e-41 & 4.85e-03\\
16&1.76e-45&2.70e-45&2.95e-45&3.04e-45&3.50e-45&4.19e-45 & 5.06e-49 & 1.66e-04 \\
17&5.76e-52&9.22e-52&1.18e-51&1.13e-51&1.33e-51&1.69e-51 & 3.40e-57 & 3.01e-06\\
18&3.82e-59& 1.25e-58&1.47e-58&1.44e-58&1.69e-58&2.18e-58&4.28e-66 &  2.97e-08 \\
19&2.37e-67&1.12e-66&4.99e-66&4.97e-66&7.98e-66&1.05e-65&9.78e-76 & 1.97e-10\\
20&4.90e-77&7.39e-75&3.41e-74&4.77e-74&8.98e-74&1.30e-73&3.95e-86& 8.28e-13 \\
21&1.61e-81&1.78e-81&1.92e-81&1.88e-81&1.96e-81&2.11e-81&2.75e-97 &1.46e-16 \\
22&4.27e-90&4.55e-90&4.81e-90&4.85e-90&5.03e-90&5.63e-90 &3.21e-109& 6.61e-20\\
23&4.13e-99&4.25e-99&4.46e-99&4.52e-99&4.67e-99&5.47e-99 &6.15e-122 &1.36e-23 \\
24 &1.31e-108&1.41e-108&1.49e-108&1.49e-108&1.55e-108&1.73e-108& 1.88e-135 & 1.26e-27  \\
25&1.59e-118&1.72e-118&1.81e-118&1.81e-118&1.85e-118&2.10e-118&9.08e-150 &5.01e-32 \\
\hline
\end{tabular}

\medskip
The results shown in Table \ref{tab:mvolumes} are based on 
20 runs of Volesti using the CB algorithm.

\end{table}
\end{document}